\documentclass[letterpaper, 12pt, journal, onecolumn, draftclsnofoot]{IEEEtran}

\usepackage{amssymb,amsmath,amsthm,color} 
\usepackage{graphicx} 
\usepackage{epsfig}
\usepackage{cite}
\usepackage[hidelinks]{hyperref}
\usepackage{bookmark}

\usepackage{algorithm}
\usepackage[noend]{algpseudocode}
\makeatletter
\newcommand{\StatexIndent}[1][3]{%
  \setlength\@tempdima{\algorithmicindent}%
  \Statex\hskip\dimexpr#1\@tempdima\relax}
\makeatother

\algnewcommand{\algorithmicgoto}{\textbf{go to }}%
\algnewcommand{\Goto}[1]{\algorithmicgoto Line~\ref{#1}}%
\algnewcommand{\Label}{\State\unskip}

\usepackage{tikz}

\usepackage{soul}

\renewcommand{\natural}{{\mathbb{N}}}

\newcommand{\real}{{\mathbb{R}}}
 
\newcommand{\union}{\cup}

\newcommand{\map}[3]{#1: #2 \rightarrow #3} 
 
\newcommand{\setdef}[2]{\{#1 \mid #2\}}

\newcommand{\GG}{\mathcal{G}}
\newcommand{\LL}{\mathcal{L}}
\newcommand{\bx}{\mathbf{x}}

\newcommand{\until}[1]{\{1,\dots,#1\}} 
\newcommand{\diag}{\operatorname{diag}} 
 
\newcommand{\EE}{\mathcal{E}}

 \newcommand{\subj}{\text{subj.\ to}}

\newcommand{\E}{\mathbb{E}}

\newcommand{\argmin}{\mathop{\operatorname{argmin}}}

\newcommand{\nbrs}{\mathcal{N}}

\newtheorem{theorem}{Theorem}[section]

\newtheorem{remark}[theorem]{Remark}

\newtheorem{assumption}[theorem]{Assumption}

\newcommand\oprocendsymbol{\hbox{$\square$}}
\newcommand\oprocend{\relax\ifmmode\else\unskip\hfill\fi\oprocendsymbol}

\graphicspath{{figs/}}

\def \synchfullname/{Partitioned Dual Decomposition}
\def \synch_name/{PDD}

\def \asynchfullname/{Asynchronous Partitioned Dual Decomposition}
\def \asynch_name/{AsynPDD}

\begin{document}

\title{Distributed Partitioned Big-Data Optimization\\ via Asynchronous Dual Decomposition}

\author{
  Ivano Notarnicola$^{1}$,~\IEEEmembership{Student Member,~IEEE,} 
  Ruggero Carli$^{2}$,~\IEEEmembership{Member,~IEEE,}
  and Giuseppe Notarstefano$^{1}$,~\IEEEmembership{Member,~IEEE}

  \thanks{
  A preliminary short version of this paper has appeared as~\cite{carli2013distributed}, 
  where only a synchronous dual scheme for the partitioned set-up was considered.
  }
  \thanks{$^{1}$Ivano Notarnicola and Giuseppe Notarstefano are with the 
    Department of Engineering, Universit\`a del Salento, Lecce, Italy, 
    \texttt{name.lastname@unisalento.it}. This result is part of a
    project that has received funding from the European Research Council (ERC)
    under the European Union's Horizon 2020 research and innovation programme
    (grant agreement No 638992 - OPT4SMART).}
    \thanks{$^{2}$Ruggero Carli is with the Department of Information 
    Engineering, University of Padova, Italy, \texttt{carlirug@dei.unipd.it}.
    }
}

\maketitle

\begin{abstract}
In this paper we consider a novel partitioned framework for distributed
optimization in peer-to-peer networks. In several important applications the
agents of a network have to solve an optimization problem with two key features:
(i) the dimension of the decision variable depends on the network size, and (ii)
cost function and constraints have a sparsity structure related to the
communication graph. For this class of problems a straightforward application of
existing consensus methods would show two inefficiencies: poor scalability and
redundancy of shared information. We propose an asynchronous distributed
algorithm, based on dual decomposition and coordinate methods, to solve
partitioned optimization problems. We show that, by exploiting the problem
structure, the solution can be partitioned among the nodes, so that each node
just stores a local copy of a portion of the decision variable (rather than a
copy of the entire decision vector) and solves a small-scale local problem.
\end{abstract}

\section{Introduction}
\label{sec:intro}
Distributed optimization has received a widespread attention in the last years
due to its key role in multi-agent systems (also known as large-scale systems,
sensor networks or peer-to-peer networks). Several solutions have been proposed,
but many challenges are still open. In this paper we focus on a main common
limitation of the current approaches. That is, in all the currently available
algorithms the nodes in the network reach consensus on the entire solution
vector. This redundancy of information may be not necessary or even realizable
in some problem set-ups. Thus, we exploit a new distributed optimization set-up
in which the nodes compute only a portion of the solution and the whole
minimizer may be obtained by stacking together the local portions.

We divide the relevant literature for our paper in two parts. That is, we review
works on distributed optimization more closely related to the techniques
proposed in this paper, and the centralized and parallel literature on big-data
optimization.

Early references on distributed optimization
are~\cite{Nedic2009distributed,Nedic2010constrained}. Convex optimization
problems are solved by using a primal distributed subgradient method combined
with a consensus scheme.
Dual decomposition methods have been proposed in early references in order to
develop distributed algorithms in a pure peer-to-peer set-up. In
\cite{yang2011distributed} a tutorial on network optimization via dual
decomposition can be found. 
In~\cite{terelius2011decentralized} a synchronous distributed algorithm
based on a dual decomposition approach is proposed for a convex optimization 
problem with a common constraint for all the agents.
In~\cite{zhu2012distributed} equality and inequality
constraints are handled in a distributed set-up based on duality. 
In~\cite{duchi2012dual} a distributed algorithm based on an averaging
scheme on the dual variables is proposed, to solve convex optimization problems
over fixed undirected networks.
A slightly
different set-up is considered in~\cite{falsone2017dual}, where a dual
decomposition method over time-varying graphs is proposed.
In order to induce robustness in the computation and improve convergence in the
case of non-strictly convex functions, Alternating Direction Methods of
Multipliers (ADMM) have been proposed in the network
context~\cite{Schizas2008consensus}. A distributed consensus optimization
algorithm based on an inexact ADMM is proposed in~\cite{chang2015multi}.
In~\cite{wei2013on} an asynchronous ADMM-based distributed method is proposed 
for a separable, constrained optimization problem.
A different class of algorithms, working under a general asynchronous
and directed communication, is based on the exchange of cutting planes among the
network nodes~\cite{burger2014polyhedral} and can be applied also in its dual
form to separable convex programs.

A common drawback of the above algorithms is that they are well suited for a
set-up in which either the dimension of the decision variable or the number of
constraints is constant with respect to the number of nodes in the network. In
case both the two features depend on the number of nodes each local computing
agent needs to handle a problem whose dimension is not scalable with respect to
the network dimension.
To cope with big-data optimization problems, deterministic and randomized coordinate 
methods for both unconstrained and constrained optimization have been proposed,
see e.g.,~\cite{bertsekas1989parallel,nesterov2012efficiency,nesterov2013gradient}.
More general set-ups such as composite and/or separable optimization in a
parallel scenario have been addressed for convex problems
in~\cite{richtarik2014iteration,richtarik2016parallel}, %
whereas nonconvex problems are considered
in~\cite{facchinei2015parallel,daneshmand2015hybrid,patrascu2015efficient}.
In~\cite{necoara2013random} an edge-based distributed algorithm is proposed 
to solve linearly coupled optimization problems via a coordinate descent method.
A distributed coordinate primal-dual asynchronous algorithm is proposed 
in~\cite{bianchi2016coordinate} to deal with large-scale problems.
A dual approach has been combined with a coordinate proximal gradient
in~\cite{notarnicola2017asynchronous} to propose an asynchronous distributed
algorithm for composite convex optimization.

In this paper we investigate a class of problems of interest in several
multi-agent applications in which the decision variable grows as the number of
nodes in the network, but the cost function and the constraints have a special
partitioned structure.
We show that such structure is not derived just as a pure academic exercise, but
vice-versa appears in several important application scenarios. In particular, we
present two of them that have been widely investigated in the literature, namely
distributed quadratic estimation and network utility maximization (and its
related resource allocation version).

The main contribution of this paper is as follows. For this problem set-up we
provide two distributed optimization algorithms, based on dual decomposition,
with two main appealing features. First, the algorithms are scalable, in the
sense that each node only processes a portion of the decision variable
vector. As a result, the information stored and the computation performed by
each node does not depend on the network size as long as the node degree is
bounded. Second, the asynchronous algorithm works under a communication protocol
in which a node wakes-up when triggered by its local timer or by
its neighbors, so that no global clock is needed.
The distributed algorithms are derived by first writing a suitable equivalent formulation 
of the original primal optimization problem (which exploits the partitioned
structure). Then its dual problem is derived and solved with suitable algorithms.
A scaled gradient applied to the dual problem turns out to be a partitioned version
of the distributed dual decomposition (synchronous) algorithm. 
A randomized ascent method applied to the dual problem allows us to write
an asynchronous distributed algorithm that converges in objective value with
high probability.

As opposed to~\cite{yang2011distributed,zhu2012distributed,duchi2012dual,terelius2011decentralized}, 
even though we also consider a dual decomposition approach, we tailored the methodology 
for the partitioned set-up, thus explicitly taking into account the partitioned 
structure of the cost and the constraints. 
This results in algorithmic formulations that reduce memory and 
communication burden.
The partitioned set-up considered in this paper has been introduced
in~\cite{erseghe2012distributed}, where a distributed ADMM algorithm is
proposed. In~\cite{erseghe2015distributed} a nonconvex maximum likelihood
localization partitioned problem is solved via a similar distributed ADMM
scheme.
In~\cite{necoara2016parallel} a convex composite optimization problem 
is considered where the cost has a partitioned part and a fully 
separable remainder. A parallel coordinate algorithm is proposed with 
its convergence analysis.
In~\cite{todescato2015robust} a robust block-Jacobi algorithm for a partitioned
quadratic programming under lossy communications is proposed.
A related formulation of the partitioned problem is the one considered 
in~\cite{mota2015distributed}, where the D-ADMM distributed algorithm 
proposed in~\cite{mota2013dadmm} has been applied.
Differently from the above references, in this paper we propose a dual
  decomposition algorithm for optimization problems in which also the
  constraints exhibit a partitioned structure.  Moreover, we develop an
  asynchronous distributed algorithm, inspired to
  \cite{notarnicola2017asynchronous}, by combining dual decomposition with
  coordinate methods.
The paper is organized as follows. In Section~\ref{sec:setup} we present the
partitioned optimization framework and describe two motivating applications. In
Section~\ref{sec:dual_derivation_and_algorithms} we develop a partitioned
distributed dual decomposition approach, then we propose and analyze our
synchronous and asynchronous distributed algorithms. Finally, in
Section~\ref{sec:simulations} we run simulations to corroborate
the theoretical results.

\section{Problem Set-up and Motivating Scenarios}
\label{sec:setup}

\subsection{Problem Set-up}
We consider a network of agents aiming at solving a structured optimization
problem in a distributed way. The nodes, $\until{n}$, interact according to
a fixed \emph{connected, undirected} graph $\GG = (\until{n}, \EE)$. We denote
$\nbrs_i$ the set of neighbors of node $i$ in $\GG$, that is
$\nbrs_i = \setdef{ j\in \until{n} }{(i,j) \in \EE}$. 
As we will see in the following, the graph $\GG$ is related to the structure of
the optimization problem.

As for the communication, we will consider a synchronous
communication protocol in which nodes communicate over the fixed graph according
to a common clock, and an asynchronous protocol in which, although the
neighboring agents are determined by the fixed graph $\GG$, communication
happens asynchronously. We will formally define this last communication protocol
in the next sections.

We start by reviewing a common set-up in distributed optimization. That is, we
consider the minimization of a separable cost function subject to local
constraints,
\begin{align}
  \begin{split}
    \min_{x\in\real^N} &\; \sum_{i=1}^n f_i(x)\\
    \subj &\; x \in X_i, \qquad i\in\until{n},
  \end{split}
  \label{eq:general_problem}
\end{align}
where $\map{f_i}{\real^N}{\real}$ and $X_i \subseteq \real^N$ for all $i\in\until{n}$.
In our set-up the local objective function $f_i$ and
the local constraint set $X_i$ are known only by agent $i$.

In this paper we want to consider problems as in~\eqref{eq:general_problem} with
a specific feature, that is a \emph{partitioned structure}, that we next
describe. Let the vector $x$ be partitioned as
\begin{align*}
  x= [ x_1^\top, \ldots, x_n^\top ]^\top
\end{align*}
where, for $i\in \until{n}$, $m_i\in \natural$, $x_i\in \real^{m_i}$ and
$\sum_{i=1}^n m_i=N$. The sub-vector $x_i$ represents the relevant information
at node $i$, hereafter referred to as the state of node $i$. Additionally, let
us assume that the local objective functions and the constraints have the same
sparsity as the communication graph, namely, for $i\in \until{n}$, the function
$f_i$ and the constraint $X_i$ depend only on the state of node $i$ and on its
neighbors, that is, on $\left\{x_j,\,j\in \nbrs_i \cup \left\{i\right\}\right\}$. 
Then the problem we aim at solving distributedly is
\begin{equation}
  \begin{split}
    \min_{x} &\; \sum_{i=1}^n f_i(x_i, \{x_j\}_{j\in \nbrs_i})\\
    \subj &\; (x_i,\{x_j\}_{j\in \nbrs_i}) \in X_i, \qquad i\in\until{n},
  \end{split}
  \label{eq:partitioned_problem}
\end{equation}
where the notation $f_i(x_i, \{x_j\}_{j\in \nbrs_i})$ means that
$\map{f_i}{\real^N}{\real}$ is in fact a function of $x_i$ and $x_j$, $j\in
\nbrs_i$, and the notation $(x_i,\{x_j\}_{j\in \nbrs_i}) \in X_i$ means that the
constraint set $X_i$ involves only the variables $x_i$ and $x_j$, $j\in \nbrs_i$.

We stress that the constraint sets $X_i$ can involve all (neighboring)
variables $(x_i,\{x_j\}_{j\in \nbrs_i})$ of agent $i$ and not just $x_i$. This
apparently minor feature in fact adds much more generality to the problem and
introduces important significant challenges. 

The following assumptions will be used in the paper. 
\begin{assumption}
  For all $i\in\until{n}$, the function $\map{f_i}{\real^{\sum_{j\in\nbrs_i \cup \{i\}} \!\!m_j}\! }{\!\!\real}$ is
  strongly convex with parameter $\sigma_i>0$.~\oprocend
  \label{ass:strong_convexity}
\end{assumption}

\begin{assumption}
  The constraint sets $X_i\subseteq \real^{\sum_{j\in\nbrs_i \cup \{i\}}m_j}$, $i\in\until{n}$ are 
  nonempty convex and compact.~\oprocend  
  \label{ass:convex_constraints}
\end{assumption}

\begin{assumption}[Constraint qualification]
  The~intersection of the relative interior of the sets $X_i$, $i\in\until{n}$,
  is non-empty.~\oprocend
  \label{ass:slater}
\end{assumption}

Under Assumptions~\ref{ass:strong_convexity} and \ref{ass:convex_constraints} 
problem~\eqref{eq:partitioned_problem} is feasible and admits a unique optimal 
solution $f^\star$ attained at some $x^\star \in \real^N$.
Assumption~\ref{ass:slater} is a standard requirement to guarantee that a dual 
approach will enjoy the strong duality property.

\subsection{Motivating Examples}
\label{sec:motivating}
Next we provide two application scenarios in which the partitioned
structure of the optimization problem arises naturally.

\subsubsection{Distributed estimation in power networks}
To describe this example we follow the treatment in \cite{pasqualetti2012}.

For a power network, the state at a certain instant of time consists of the
voltage angles and magnitudes at all the system buses. The (static) state
estimation problem refers to the procedure of estimating the state of a power
network given a set of measurements of the network variables, such as, for
instance, voltages, currents, and power flows along the transmission lines. To
be more precise, let $x \in \real^N$ and $z \in \real^P$ be, respectively, the state and
measurements vector.  Then, the vectors $x$ and $z$ are related by the relation
\begin{equation}
  z = h(x) + \eta,
  \label{eq:Measure}
\end{equation}
where $h(\cdot)$ is a nonlinear measurement function, and where $\eta$ is the noise
measurement, which is traditionally assumed to be a zero mean random vector
satisfying $\E[\eta \eta^\top] = \Sigma \succ 0$. An optimal estimate of the network
state coincides with the most likely vector $x^\star$ that solves equation
\eqref{eq:Measure}. This static state estimation problem can be simplified by
adopting the approximated estimation model presented in \cite{schweppe70}, which
follows from the linearization around the origin of equation~\eqref{eq:Measure}. 
Specifically,
\begin{align*}
  z = Hx + v,
\end{align*}
where $H \in \real^{P \times N}$ and where $v$, the noise measurement, is such that
$\E[v] = 0$ and $\E[vv^\top] = \Sigma$. In this context the static state estimation
problem is formulated as the following weighted least-squares problem
\begin{equation}
  \argmin_{x} \: (z - Hx)^\top \Sigma^{-1} (z - Hx).
  \label{eq:wlsPowerNetworks}
\end{equation}
Assume $\ker(H) = {0}$, then the optimal solution to the above problem is given
by
\begin{align*}
  x_{\text{wls}} = \left( H^\top \Sigma^{-1} H \right)^{-1} H^\top \Sigma^{-1}z.
\end{align*}
For simplicity let us assume that $\Sigma=I$.  For a large power network,
the centralized computation of $x_{\text{wls}}$ might be too onerous. A possible
solution to address this complexity problem is to distribute the computation of
$x_{\text{wls}}$ among geographically deployed control centers (monitors), say $n$ in a way
that each monitor is responsible only for a subpart of the whole
network. Precisely let the matrices $H$ and $\Sigma$ and the vector $z$ be
partitioned as $\left[H_{ij}\right]_{i,j=1}^n$, $x=\left[ x_1^\top, \ldots, x_n^\top \right]^\top$
and $z=\left[z_1^\top, \ldots, z_n^\top \right]^\top$, where 
$H_{ij} \in\real^{p_i \times m_j}$, $z_i\in  \real^{p_i}$, $x_i\in \real^{m_i}$ and 
$\sum_{i=1}^n m_i = N$, $\sum_{i=1}^n p_i = P$. Observe that, because of the interconnection
structure of a power network, the measurement matrix $H$ is usually sparse, i.e.,
many $H_{ij}=0$.  Assume monitor $i$ knows $z_i$ and $H_{ij}$, $j\in \until{n}$
and it is interested only in estimating the sub-state $x_i$. Moreover let
$\nbrs_i=\left\{j \in \until{n} \mid H_{ij}\neq 0\right\}$. Observe that in general
if $H_{ij}\neq 0$ then also $H_{ji}\neq 0$. Then by defining
\begin{align*}
  f_i \big( x_i, \{x_j\}_{j\in \nbrs_i} \big) = 
  \Big( z_i-\sum_{j\in \nbrs_i} H_{ij}x_j \Big)^{\!\top}
  \Big( z_i-\sum_{j\in \nbrs_i} H_{ij}x_j \Big),
\end{align*}
problem~\eqref{eq:wlsPowerNetworks} can be equivalently rewritten as
\begin{align*}
  \argmin_x \: \sum_{i=1}^n f_i \big( x_i, \{x_j\}_{j\in \nbrs_i} \big)
\end{align*}
which is of the form~\eqref{eq:partitioned_problem}. 

It is worth remarking that there are other significant examples that can be 
cast as distributed weighted least square problems similarly to the static 
state estimation in power networks we have described in this section; see, 
for instance, distributed localization in sensor networks and map building 
in robotic networks.

\subsubsection{Network utility maximization and resource allocation}
We consider the flow optimization problem, or \emph{Network Utility Maximization
(NUM)} problem introduced in \cite{kelly1998rate} and studied in~\cite{low1999optimization} 
in a distributed context. A flow network (which is
different from a communication network) consists of a set $L$ of unidirectional
links with capacities $c_\ell$, $\ell \in L$. The network is shared by a set of $n$
sources.
Each source has a strongly concave utility function $U_i(x_i)$
The goal is to calculate source
rates that maximize the sum of the utilities $\sum_{i=1}^n U_i(x_i)$ over $x_i$
subject to capacity constraints. Formally, using a notation consistent with
\cite{low1999optimization}, let $L(i) \subseteq L$ be the set of links used by
source $i$ and $N(\ell) = \{i \in \until{n} \mid \ell \in L(i)\}$ be the set of
sources that use link $\ell$. Note that $\ell \in L(i)$ if and only if
$i \in N(\ell)$.  Also, let $I_i = [\kappa_i, K_i]$, with $0\le \kappa_i<K_i$,
be the interval of transmission rates allowed to node $i$. The network flow
optimization problem is given by
\begin{equation}
  \label{eq:network_flow_optimization}
  \begin{split}
    \max_{x_1,\ldots,x_n} \: & \: \sum_{i=1}^n U_i(x_i) 
    \\
    \subj \: & \: x_i \in I_i,\hspace{1.9cm} i \in \until{n},
    \\
    & \sum_{j \in N(\ell)} x_j \leq c_\ell, \hspace{0.85cm} \ell \in \until{ | L | }.
  \end{split}
\end{equation}
Notice that problem~\eqref{eq:network_flow_optimization} is well posed and has
compact domain.

In Figure~\ref{fig:NUM} (left) we graphically represent an example of $5$
sources (filled circles) that use (dotted arrows) $3$ links (gray stripes). 
In~\cite{low1999optimization} a distributed optimization algorithm is proposed
in which both the sources and the links are computation units.
Here we consider a set-up in which only the sources are computation units. In
particular, the sources have the computation and communication capabilities
introduced in the previous subsection. We assume that sources using the same
links can communicate and both know the capacity constraint
on those links. Formally, we introduce a graph $\GG$ having
an edge $(i,j)$ connecting source $i$ to $j$ if and only if there exists 
$\ell \in L$ such that $\ell \in L(i)\cap L(j)$.
In Figure~\ref{fig:NUM} (right) we show the induced communication graph 
(solid lines) for the considered example.
\begin{figure}[!ht]
\centering
  \includegraphics[scale=0.8]{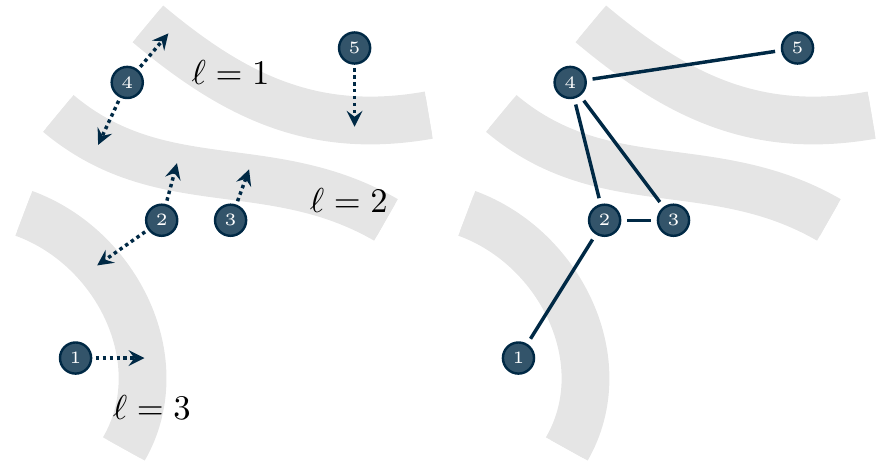}
  \caption{
  Network Utility Maximization problem with $5$ sources (filled circles) using $3$ links 
  (gray stripes).
  }
\label{fig:NUM}
\end{figure}

Thus, optimization problem~\eqref{eq:network_flow_optimization} can be rewritten as%
\begin{align}
  \begin{split}
    \max_{x_1,\ldots,x_n} \: & \: \sum_{i=1}^n U_i(x_i) 
    \\
    \subj \: & \:\:  (x_i,\{x_j\}_{j\in \nbrs_i}) \in
    \prod_{j\in\nbrs_i \cup \{i\}} I_j, \: i\in\until{n},
    \\
    & \sum_{j \in \nbrs_i \cup \{i \} } \!\! a_{j,\ell} \, x_j \leq c_\ell, \hspace{0.2cm} \ell \in L(i),\:i\in\until{n},
  \end{split}
  \label{eq:network_flow_optimization_partitioned}
\end{align}
where $\nbrs_i$ is the neighbor set of $i$ in $\GG$, 
$a_{j,\ell} $ is $1$ if agent $j$ can use link $\ell$ 
and $0$ otherwise.
Notice that in problem~\eqref{eq:network_flow_optimization_partitioned} if
sources $i$ and $j$ share a link $\ell$, then they both have the capacity
constraint of link $\ell$. Moreover, in order to have compactness of
the local constraint set $X_i$, transmission rate constraints of neighboring
nodes are also taken into account.
Finally, in order to fulfill the strong convexity assumption on the local costs,
two strategies can be used. First, one can assume that each agent
knows the utility functions of neighboring agents, so that it sets
  $f_i ( x_i,\{x_j \}_{j\in\nbrs_i} ) = \sum_{j\in\nbrs_i \cup \{i\} } U_j (x_j )$.
Alternatively, one can consider an additional separable (small) regularization term in the NUM problem formulation
in~\eqref{eq:network_flow_optimization}, e.g., $\epsilon \sum_{i=1}^n x_i^2$ with $\epsilon>0$. 
In this case each agent sets its local cost function to
  $f_i ( x_i,\{x_j \}_{j\in\nbrs_i} ) = U_i (x_i ) + \sum_{j\in\nbrs_i } \epsilon_{ij} x_j^2$,
where $\epsilon_{ij}>0$ are suitable fractions of $\epsilon$.
Except for the maximization versus minimization, this problem 
is partitioned, that is it has the same structure 
as~\eqref{eq:partitioned_problem}.

A problem with this structure can be also found in resource allocation problems,
which are of great importance in several research areas. In the context of
network systems solving resource allocation problems in a distributed way is a
preliminary task to solve several control and estimation problems. Indeed, it is
often the case that the agents in the network have some local resource that have
to share with their neighbors.

Consider a general set-up in which each agent produces a certain amount of
resource, which it can share with its neighbors (i.e., neighboring nodes in the
communication graph). Each agent has a local strongly concave utility function 
to maximize. The resulting optimization problem turns out to be
\begin{align*}
  \max_{x_1,\ldots,x_n} \: & \: \sum_{i=1}^n U_i (x_i)
  \\
  \subj \: &\: \sum_{j\in \nbrs_i \union \{ i \} } x_j \leq r_i, \qquad i\in\until{n},
\end{align*}
where $x_i$ is the resource produced by node $i$, $r_i$ is the capacity of node
$i$ and we are assuming that the set of neighbors with whom node $i$ can share
its resource coincides with the set of neighbors in the communication graph. In
other words agents can share resources only if they can communicate.

It is worth noting that dual decomposition is often used in network utility
maximization and resource allocation problems. See for example
\cite{palomar2006tutorial} for a tutorial on dual decomposition methods in
network utility maximization. Usually, in this context, the capacity constraints
are dualized to obtain a master-subproblem or a distributed algorithm. However,
in these early references, as e.g., in \cite{low1999optimization}, the dual
decomposition gives rise to algorithms that are not suited for a pure
peer-to-peer network as the one we consider.  In our partitioned approach
the dual decomposition is used to enforce the coherence constraints, whereas
the capacity constraints are taken into account in the primal local
minimization. These aspects will be more clear in the next section in which we derive 
the partitioned dual decomposition approach.

\section{Partitioned Dual Decomposition for Distributed Optimization}
\label{sec:dual_derivation_and_algorithms}

In order to introduce our distributed algorithms, we derive a partitioned
dual decomposition scheme by introducing suitable copies of the decision
variables.

As a preliminary step, we briefly recall the standard dual decomposition approach for
distributed optimization.
In order to solve problem~\eqref{eq:general_problem} in a distributed way, a
common approach consists of writing it in the equivalent form
\begin{equation}
  \begin{split}
    \min_{x^{(1)}, \ldots, x^{(n)}} \: &\: \sum_{i=1}^n f_i(x^{(i)})
    \\
    \subj \: &\: x^{(i)} \in X_i,  \hspace{0.7cm} i\in\until{n},
    \\
    &\: x^{(i)} = x^{(j)}, \hspace{0.5cm} (i,j)\in \EE,
  \end{split}
  \label{eq:general_problem_copies}
\end{equation}
where each $x^{(i)}$ can be seen as a copy of $x$ subject to the additional
constraint that all the copies must be equal. Clearly, the connected nature 
of the network ensures equality between all $x^{(i)}$ and, in turn, the 
equivalence between~\eqref{eq:general_problem_copies} and~\eqref{eq:general_problem}.

When considering a partitioned problem as in \eqref{eq:partitioned_problem},
because of the structure of $f_i$ and $X_i$, $i \in \until{n}$, the
formulation~\eqref{eq:general_problem} is considerably redundant.
The idea is to exploit the partitioned structure to modify
\eqref{eq:general_problem_copies} in order to limit the range of equivalences
among the auxiliary variables, and, in turn, their diffusion over the network.

\subsection{Partitioned Dual Decomposition Set-up}
\label{sec:dual_dec_setup}
Once we create copies of the vector $x\in\real^N$, we enforce each state $x_i\in\real^{m_i}$ 
to be identical only for the neighboring nodes $j\in \nbrs_i \cup \{ i \}$
which use this information. Formally, we reformulate
problem~\eqref{eq:partitioned_problem} as
\begin{equation}
  \begin{split}
    \min \: & \: \sum_{i=1}^n 
      f_i \big( x^{(i)}_i, \{x^{(i)}_j\}_{j\in \nbrs_i} \big) 
      \\
    \subj \: & \: \big( x^{(i)}_i,\{x^{(i)}_j\}_{j\in \nbrs_i} \big) \in X_i \qquad i\in\until{n},
      \\
      & \: x^{(i)}_i = x^{(j)}_i, \qquad\qquad j\in \nbrs_i, \; i\in\until{n},
      \\
      & \: x^{(i)}_j = x^{(j)}_j, \qquad\qquad j\in \nbrs_i, \; i\in\until{n},
  \end{split}
  \label{eq:partitioned_problem_copies}
\end{equation}
where $x^{(j)}_i$ denotes the copy of state $x_i$ stored in memory of node $j$. 
Notice that connected nature of the graph $\GG$ ensures equivalence
between~\eqref{eq:partitioned_problem} and~\eqref{eq:partitioned_problem_copies}.

  As an example, in Figure~\ref{fig:partitions} we visualize the partitioned
  set-up for a path graph of $n=4$ nodes.  Along $i$-th column, we show the
  coupling due to the local cost $f_i$ and the local constraint $X_i$, which
  involves only the states handled by node $i$, i.e., $x_i^{(i)}$ and
  $x_j^{(i)}$ with $j\in\nbrs_i$. Along the $i$-th row, we show the coupling due
  to copies $x_i^{(j)}$, $j\in\nbrs_i$, of the variable $x_i$. 
\begin{figure}[!ht]
  \centering
\hspace{1cm}
\includegraphics[scale=0.9]{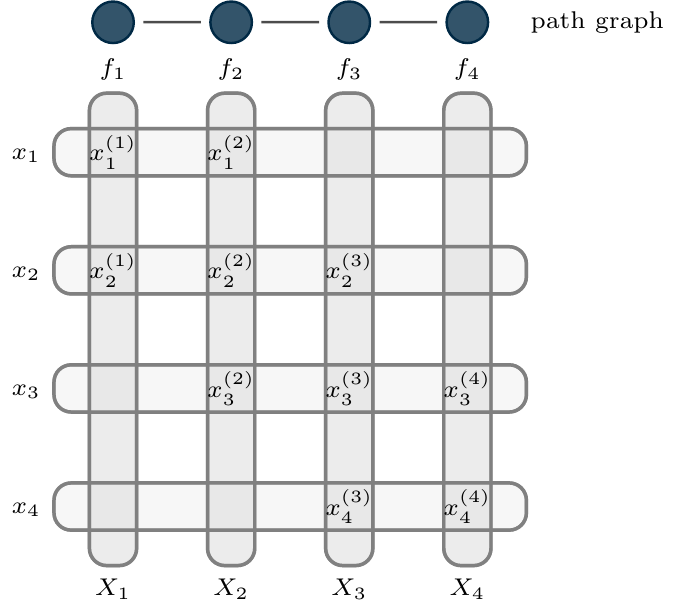}
\caption{
  Partitioned optimization problem over a path graph of $n = 4$ nodes.
  }
\label{fig:partitions}
\end{figure}

Before proceeding with presentation of the algorithms, we discuss two key
features in the structure of the above problem.
First, it is worth noting that the problem formulations
\eqref{eq:general_problem_copies} and \eqref{eq:partitioned_problem_copies},
although equivalent, are different. In fact,
\eqref{eq:partitioned_problem_copies} will lead to our partitioned
algorithm.
Second, we point out that a constraint $x^{(i)}_i = x^{(j)}_i$ for a pair of
agents $i$ and $j$ appears two times. This redundant formulation is not
accidental, but plays an important role in exploiting the partitioned structure
of the proposed algorithm.

Next, we introduce an aggregate notation for
the copies, which allows us to be more compact in the derivation of the
algorithms and their analysis. We denote by
\begin{align}
  y^{(i)} := \big( x^{(i)}_i, \{x^{(i)}_j\}_{j\in \nbrs_i} \big)
  \label{eq:yi_definition}
\end{align}
the set of local variables of node $i$, arranged as a column vector in 
$\real^{\sum_{j\in\nbrs_i \cup \{ i\} }m_j}$. In this way we can write 
equivalently
\begin{align*}
  f_i \big( x^{(i)}_i, \{x^{(i)}_j\}_{j\in \nbrs_i} \big) = f_i \big( y^{(i)}
  \big) \quad \text{and} \quad y^{(i)}\in X_i.
\end{align*}

To tackle problem~\eqref{eq:partitioned_problem_copies} in a distributed way, we 
start by deriving its dual problem. 
The \emph{partial} Lagrangian for problem
\eqref{eq:partitioned_problem_copies} is given by
\begin{align}
  \label{eq:partial_lagrangian}
    \LL({\bx,\Lambda}) & = \sum_{i=1}^n \bigg ( 
      f_i \big( x^{(i)}_i, \{x^{(i)}_j\}_{j\in \nbrs_i} \big)
    \\ \notag &
    +\!\! \sum_{j\in \nbrs_i} \Big ( \lambda^{(i,j) \top}_i (x^{(i)}_i - x^{(j)}_i)
        + \! \lambda^{(i,j) \top}_j (x^{(i)}_j - x^{(j)}_j) \! \Big )\! \bigg ),
\end{align}
where $\bx$ stacks all the (primal) optimization variables in the network, while $\Lambda$ 
denotes the stack of dual variables, i.e.,
\begin{align*}
  \Lambda = \big[ \Lambda_1^\top, \ldots, \Lambda_n^\top \big] ^\top,
\end{align*}
with block $\Lambda_i := [ \{\lambda_i^{(i,j)} \}_{j\in\nbrs_i}, \{\lambda_j^{(i,j)} \}_{j\in\nbrs_i} ]$, 
$i\in\until{n}$.

By exploiting the undirected nature and the connectivity of graph $\GG$, the
Lagrangian~\eqref{eq:partial_lagrangian} can be rewritten as
\begin{align}
  \label{eq:partial_lagrangian_separable}
  \LL( \bx,\Lambda & )  = \sum_{i=1}^n \bigg( f_i \big( x^{(i)}_i, \{x^{(i)}_j\}_{j\in \nbrs_i} \big)
  \\ \notag
  & \!\! + x_i^{(i) \top} \!\! \sum_{j\in \nbrs_i} (\lambda^{(i,j)}_i \!-\! \lambda^{(j,i)}_i) 
    +\!\! \sum_{j\in \nbrs_i} x_j^{(i) \top} \big( \lambda^{(i,j)}_j \!-\! \lambda^{(j,i)}_j \big) \!\bigg)
\end{align}
which is separable with respect to $y^{(i)}$, $i\in\until{n}$.

\begin{remark}
  It is worth noting that we have not dualized the local constraints
  $(x^{(i)}_i,\{x^{(i)}_j\}_{j\in \nbrs_i}) \in X_i$ (thus the notion of partial
  Lagrangian) since each of them will be handled by the agents in their local
  optimization problem.~\oprocend
\end{remark}

The dual function of~\eqref{eq:partitioned_problem_copies} is obtained by minimizing
the Lagrangian with respect to the primal variables, which gives
\begin{align*}
\begin{split}
  q( \Lambda ) & = \min_{\bx\in X_1\times \cdots \times X_n} \LL({\bx,\Lambda}) 
  \\
  & = \sum_{i=1}^n q_i \Big( 
    \big \{ \lambda^{(i,j)}_i,\lambda^{(j,i)}_i,\lambda^{(i,j)}_j,\lambda^{(j,i)}_j \big \}_{j\in\nbrs_i}
  \Big)
\end{split}
\end{align*}
with
\begin{align}
  \notag
  & q_i \Big( \big \{
  \lambda^{(i,j)}_i,\lambda^{(j,i)}_i,\lambda^{(i,j)}_j,\lambda^{(j,i)}_j
  \big \}_{j\in\nbrs_i} \Big) =\\ 
  &\hspace{0.5cm} \min_{ \big( x^{(i)}_i, \{x^{(i)}_j\}_{j\in \nbrs_i}  \big ) \in X_i} \!\!
    \Big( f_i \big( x^{(i)}_i, \{x^{(i)}_j\}_{j\in \nbrs_i} \big)
  \label{eq:qi_defintion}
  \\ 
  \notag
  &\hspace{0.9cm} + x_i^{(i) \top} \!\! \sum_{j\in \nbrs_i} \!\! (\lambda^{(i,j)}_i \!-\! \lambda^{(j,i)}_i) 
    +\!\! \sum_{j\in \nbrs_i} \!\! x_j^{(i) \top} (\lambda^{(i,j)}_j \!-\! \lambda^{(j,i)}_j)\Big).
\end{align}

  Notice that, since each $X_i$ is compact and nonempty, the minimum
  in~\eqref{eq:qi_defintion} is (uniquely) attained, so that $q_i$ is always
  finite. Thus, the dual problem of~\eqref{eq:partitioned_problem_copies} is
the following unconstrained optimization problem
\begin{align}
  \max_{\Lambda} \: & \: \sum_{i=1}^n q_i 
    \Big( \big \{ \lambda^{(i,j)}_i,\lambda^{(j,i)}_i,\lambda^{(i,j)}_j,\lambda^{(j,i)}_j \big \}_{j\in\nbrs_i}\Big).
\label{eq:dual_problem}
\end{align}

\begin{remark}
  Let $\map{\varphi}{\real^d}{\real \cup \{+\infty \}}$, its conjugate function 
  $\map{\varphi^*}{\real^d}{\real}$ is defined as 
  \begin{align*}
    \varphi^*(z) := \sup_x \big( z^\top x - \varphi(x) \big).
  \end{align*}
  Then,
  \begin{align*}
    & q_i\Big( \big \{
    \lambda^{(i,j)}_i,\lambda^{(j,i)}_i,\lambda^{(i,j)}_j,\lambda^{(j,i)}_j
    \big \}_{j\in\nbrs_i}\Big) = \\[1.2ex]
   &\hspace{1.4cm} -f_i^*\Big( \sum_{j\in \nbrs_i} \!\! (\lambda^{(i,j)}_i \!-\! \lambda^{(j,i)}_i) , 
   \big \{ (\lambda^{(i,j)}_j \!-\! \lambda^{(j,i)}_j) \big \}_{j \in \nbrs_i} \Big),
  \end{align*}  
  with $f_i^*$ being the conjugate function of $f_i$.~\oprocend
\label{rem:qi_conjugate}
\end{remark}

\begin{remark}
  It is worth noting that each $q_i$ does not depend on the entire set of dual
  variables $\Lambda$, but it exhibits a sparse structure, i.e., it is a function 
  of the dual variables of the neighbors $\nbrs_i$ only.~\oprocend
\end{remark}

\subsection{(Synchronous) Partitioned Dual Decomposition (\synch_name/) distributed algorithm}
\label{sec:PBDD}
With the dual problem in hand, a gradient algorithm on the dual problem,
\cite[Chapter 6]{bertsekas1999nonlinear}, can be applied. This results into a
minimization on the primal variables and a linear update on the dual
variables. As we will show in the analysis, this gives rise to the \synch_name/
distributed algorithm, which is formally stated, from the perspective of node
$i$, in the following table.

We point out that each node $i\in\until{n}$ stores and updates the primal
variables $x^{(i)}_i$ and $x^{(i)}_j$, $j\in \nbrs_i$, and the dual variables
$\lambda^{(i,j)}_i$ and $\lambda^{(i,j)}_j$, $j\in \nbrs_i$.

\begin{algorithm}[H]
  \floatname{algorithm}{Distributed Algorithm}
  \begin{algorithmic}[0] 
    \Statex\textbf{Processor states}:
    $( x_i^{(i)}, \{ x_j^{(i)} \}_{j\in\nbrs_i}  )$ and
    $\{ \lambda_i^{(i,j)}, \lambda_j^{(i,j)}\}_{j\in \nbrs_i}$\smallskip
    \Statex \textbf{Evolution}:
    \StatexIndent[0.25] \textsc{for}: $t=1,2,\ldots$ \textsc{do}

    \StatexIndent[0.5] Compute and broadcast primal variables
    \begin{align}
    \begin{split}    
      & \big( x^{(i)}_i (t \!+\! 1),\{ x^{(i)}_j (t \!+\! 1) \}_{j\in \nbrs_i}\big)=
      \\
      & \hspace{0.5cm} \argmin_{(x_i, \{x_j\}_{j\in \nbrs_i}) \in X_i}
      \bigg( f_i \big( x_i, \{x_j\}_{j\in \nbrs_i} \big)
        + x_i^\top  \textstyle\sum\limits_{j\in \nbrs_i} \! \Big( \lambda^{(i,j)}_i (t) - \lambda^{(j,i)}_i (t) \Big) 
        + \textstyle\sum\limits_{j\in \nbrs_i} \! x_j^\top \Big( \lambda^{(i,j)}_j (t) - \lambda^{(j,i)}_j (t) \Big)  \!\! \bigg).
    \end{split}
    \label{eq:local_minimization}
   \end{align}

    \StatexIndent[0.5] Update and broadcast dual variables via
    \begin{equation}
    \label{eq:synch_ascent}
    \begin{split}
      & \lambda^{(i,j)}_i (t \!+\! 1) = \lambda^{(i,j)}_i(t) \!+\! \alpha_i \Big(\! x^{(i)}_i (t \!+\! 1) - x^{(j)}_i (t \!+\! 1) \! \Big)
      \\
      & \lambda^{(i,j)}_j (t \!+\! 1) = \lambda^{(i,j)}_j(t) \!+\! \alpha_i \Big( \! x^{(i)}_j (t \!+\! 1) - x^{(j)}_j (t \!+\! 1) \!\Big)
    \end{split}
    \end{equation}    
    \StatexIndent[0.5] for all $j\in\nbrs_i$.
  \end{algorithmic}
  \caption{\synch_name/} 
  \label{alg:PBDD}
\end{algorithm}

Before studying the convergence properties of the proposed algorithm, let us comment on
its scalability and how it compares with standard dual gradients algorithms.
First, observe that each node has to keep in memory the set of variables
$x_i^{(i)}$, $\big \{ x_j^{(i)} \big \}_{j \in \nbrs_i}$, $ \big\{\lambda^{(i,j)}_i,
  \lambda^{(i,j)}_j \big\}_{j \in \nbrs_i}$, namely a number of variables equal to
$1+3 |\nbrs_i|$.
Second, the step-sizes $\alpha_i$, $i\in\until{n}$ are constant, local and can be 
initialized via local computations. More details are given in
Theorem~\ref{thm:synch_convergence}.

\begin{remark}
  Notice that, differently from existing dual decomposition schemes,
  our algorithms do not enforce any symmetry in the dual variables,
  i.e., in general $\lambda^{(i,j)}_i(t) \neq - \lambda^{(j,i)}_i(t)$.
  The symmetry, although not necessary, can be imposed if the agents
  select a common step-size $\alpha_i = \alpha$, for all $i\in\until{n}$,
  and properly initialize their dual variables. As a consequence, the
  algorithm can be simplified to have only one communication round to 
  perform both the local minimization and the ascent.\oprocend
\end{remark}

The convergence properties of \synch_name/ (Distributed Algorithm~\ref{alg:PBDD}) are 
established in the following theorem.
\begin{theorem}
  Let Assumptions~\ref{ass:strong_convexity}, \ref{ass:convex_constraints} and \ref{ass:slater}
  hold true and assume the step-sizes $\alpha_i$, $i\in\until{n}$,
  to be constant and such that $0 <\alpha_i \le \frac{1}{nL_i}$, with
  \begin{align}
    L_i = \sqrt{ 2 \textstyle\sum\limits_{j\in \nbrs_i} 
        \Big(\frac{1}{\sigma_i} + \frac{1}{\sigma_j}\Big)^2}, \quad \forall\, i \in \until{n}.
  \label{eq:L_i_synch_definition}
  \end{align}
  Then, the sequence $\{ \Lambda_1 (t), \ldots, \Lambda_n (t) \}$ generated by
  \synch_name/ (Distributed Algorithm~\ref{alg:PBDD}) converges in objective
  value to the optimal cost $f^\star$ of
  problem~\eqref{eq:partitioned_problem}. Moreover, let
  $x^\star= (x_1^{\star\top},\ldots, x_n^{\star\top})^\top$ be the unique
  optimal solution of~\eqref{eq:partitioned_problem}, then each primal sequence
  $x_i^{(i)}(t)$ generated by \synch_name/ is such that
  \begin{align*}
    \lim_{t\to\infty} \, \| x_i^{(k)}(t) - x_i^\star\| = 0,
  \end{align*}%
  for all $i\in \until{n}$ and $k\in\{i\}\cup\nbrs_i$.
\label{thm:synch_convergence}
\end{theorem}%
\begin{proof}
  We structure the proof of the first statement in three parts in which we show
  that: (i) the dual gradient has a block structure and smoothness, (ii) the
  distributed algorithm implements a diagonally-scaled gradient method, and
  (iii) strong duality holds.
  First, consider the dual problem~\eqref{eq:dual_problem} and a block
  partitioning of dual variables $\Lambda = [\Lambda_1,\ldots,\Lambda_n]$, with
\begin{align}
  \Lambda_i := 
    \Big( \{ \lambda_i^{(i,j)} \}_{j\in\nbrs_i}, \{ \lambda_j^{(i,j)}{j\in\nbrs_i}  \} \Big)
\label{eq:lambda_i_definition}
\end{align}
representing the local variables of node $i$, for all $i\in\until{n}$.
Under Assumption~\ref{ass:strong_convexity}, the dual function $q(\Lambda)$
is guaranteed to have block-coordinate Lipschitz continuous gradient $\nabla q(\Lambda)$ 
with block constants $L_i$, $i\in\until{n}$, given in~\eqref{eq:L_i_synch_definition}. 
In fact, we can explicitly compute the components of $\nabla q(\Lambda)$ associated 
to each block $\Lambda_i$, denoted hereafter as $\nabla_{\Lambda_i} q(\Lambda)$, 
by using the chain rule of derivation and the conjugate function notation. 
We have that
\begin{align}
  \begin{split}
    \frac{\partial q (\Lambda)}{\partial \lambda_i^{(i,j)}} & = (\nabla f_i^*)_i - (\nabla f_j^*)_i, \: j\in\nbrs_i
    \\
    \frac{\partial q (\Lambda)}{\partial \lambda_j^{(i,j)}} & = (\nabla f_i^*)_j - (\nabla f_j^*)_j, \: j\in\nbrs_i,
  \end{split}
  \label{eq:dual_gradient_conjugate}
\end{align}
where $(\nabla f_i^*)_i$ denotes the $i$-th component of $\nabla f_i^*$ and we 
omit the argument of $\nabla f_i^*$ to take light the notation.
Since for all $i\in\until{n}$, each $f_i$ is a strongly convex function, then
the gradient of its conjugate function $\nabla f_i^*$ is Lipschitz continuous
with constant $1/\sigma_i$, \cite[Chapter~X, Theorem~4.2.2]{hiriart1993convex}.
By considering the Euclidean $2$-norm, in light
of~\eqref{eq:dual_gradient_conjugate} and by simple algebraic manipulation, we
can conclude that also $\nabla_{\Lambda_i} q(\Lambda)$ is Lipschitz continuous
with constant
  \begin{align*}
    L_i = \sqrt{ \textstyle\sum\limits_{j\in \nbrs_i} \Big(\frac{1}{\sigma_i} + 
    \frac{1}{\sigma_j}\Big)^2 + \textstyle\sum\limits_{j\in \nbrs_i} \Big(\frac{1}{\sigma_i} + \frac{1}{\sigma_j}\Big)^2},
  \end{align*}
  which matches~\eqref{eq:L_i_synch_definition}.
  Second, we show that our \synch_name/ distributed algorithm implements a
  scaled gradient ascent method to solve problem~\eqref{eq:dual_problem}.
  Consider a diagonal positive definite matrix defined as
  $W:=\diag (\alpha_1,\ldots, \alpha_n) \preceq \diag ( \frac{1}{nL_1},\ldots,
  \frac{1}{nL_n})$. Formally, the scaled gradient ascent method can be written as
  \begin{align}
    \Lambda(t+1) = \Lambda(t) + W \nabla q(\Lambda(t)),
  \label{eq:scaled_gradient_ascent}
  \end{align}
  where $t$ denotes the iteration counter. 
  Since each entry of the scaling
  matrix satisfies $0<\alpha_i\le \frac{1}{nL_i}$ for all $i\in\until{n}$, then
  the following condition holds~\cite[Theorem~8]{richtarik2016parallel}
  \begin{align*}
  \setlength{\arraycolsep}{1pt}
  q(\Lambda(t) + \delta )  \ge  q(\Lambda(t)) + \nabla q(\Lambda(t))^{\!\top} \delta 
  - \frac{n}{2} \delta^{\! \top \!} \!
  \begin{bmatrix}
    L_1 \\[-0.2cm] & \ddots \\[-0.05cm] & & L_n
  \end{bmatrix} \!
  \delta,
  \end{align*}
  for every perturbation $\delta$.
  Thus, using the same line of proof of the gradient algorithm~\cite[Chapter~2]{bertsekas1999nonlinear},
  we can conclude that the sequence $\{\Lambda(t)\}$ generated by iteration~\eqref{eq:scaled_gradient_ascent} 
  converges in objective value to the optimal cost $q^\star$ of~\eqref{eq:dual_problem}.
  Since $W$ is diagonal, then~\eqref{eq:scaled_gradient_ascent}
  splits in a component-wise fashion giving
  \begin{align}
    \Lambda_i(t+1) = \Lambda_i(t) + W_{ii} \nabla_{\Lambda_i} q(\Lambda(t)), \: i\in\until{n},
  \label{eq:scaled_gradient_ascent_block}
  \end{align}
  where $W_{ii}$ denotes the $(i,i)$-th entry of $W$.
  By using the following property of conjugate functions
  \begin{align*}
    \nabla \varphi^*(z) = \argmin_x \big( \varphi (x) - z^\top x \big),
  \end{align*}
  we have that the primal minimization~\eqref{eq:local_minimization} 
  computes $\nabla f_i^*$ evaluated at the point
  $\big( \sum_{j\in \nbrs_i} (\lambda^{(i,j)}_i(t) -
  \lambda^{(j,i)}_i(t)) , \{ (\lambda^{(i,j)}_j(t) - \lambda^{(j,i)}_j(t))\}_{j \in \nbrs_i} \big)$.
  Then
  \begin{align}
  \begin{split}
    \frac{\partial q (\Lambda(t))}{\partial \lambda_i^{(i,j)}} & = x_i^{(i)}(t+1)- x_i^{(j)}(t+1), \: j\in\nbrs_i
    \\
    \frac{\partial q (\Lambda(t))}{\partial \lambda_j^{(i,j)}} & = x_j^{(i)}(t+1)  - x_j^{(j)}(t+1), \: j\in\nbrs_i,
  \end{split}
  \label{eq:dual_gradient_xixj}
  \end{align}
  so that update~\eqref{eq:synch_ascent} is the scaled gradient 
  ascent~\eqref{eq:scaled_gradient_ascent_block}.
  Third and final, by Assumption~\ref{ass:slater} (Slater's condition), strong
  duality between problems~\eqref{eq:partitioned_problem_copies}
  and~\eqref{eq:dual_problem} holds. Moreover, since
  problems~\eqref{eq:partitioned_problem_copies}
  and~\eqref{eq:partitioned_problem} are equivalent, then they both have optimal
  cost $q^\star=f^\star$. Thus, the sequence $\{\Lambda(t)\}$ generated by
  \synch_name/ converges in objective value to the optimal cost $f^\star$
  of~\eqref{eq:partitioned_problem}.

  For the second part of the statement, we first notice that in light of
  Assumptions~\ref{ass:strong_convexity} and \ref{ass:convex_constraints},
  problem~\eqref{eq:partitioned_problem} has a unique optimal solution
  $x^\star= (x_1^{\star\top},\ldots, x_n^{\star\top})^\top$.
  Further, since problem~\eqref{eq:partitioned_problem_copies} is equivalent to
  problem~\eqref{eq:partitioned_problem}, then $x^\star$ is the unique optimal
  solution also for problem~\eqref{eq:partitioned_problem_copies}.
  Finally, the first order optimality condition for the (unconstrained) dual
  problem~\eqref{eq:dual_problem} is $\nabla q (\Lambda^\star) = 0$, where
  $\Lambda^\star$ is a limit point of the sequence $\{\Lambda(t)\}$ (which
  exists by the Lipschitz continuity of $\nabla q(\Lambda)$). This allows us to
  conclude, by equation~\eqref{eq:dual_gradient_xixj}, that the limit point of
  the primal sequences $\{ x_i^{(i)}(t), \{ x_j^{(i)} \}_{j\in\nbrs_i}(t) \}$
  satisfy the primal coherence constraints. Thus, in the limit the copies
  $x_i^{(i)}, \{ x_i^{(j)} \}_{j\in\nbrs_i} $ of the variable $x_i$ are equal to
  the (unique) optimal $x_i^\star$. Iterating on $i\in\until{n}$ the proof
  follows.
\end{proof}

\begin{remark}
  Alternative expressions for $L_i$ in~\eqref{eq:L_i_synch_definition} can be
  used. Larger upper bounds on the step-sizes $\alpha_i$ can be established by
  exploiting tailored descent conditions. See, e.g.,
  works~\cite{nesterov2012efficiency,richtarik2016parallel,necoara2016parallel}.~\oprocend
\end{remark}

\subsection{Asynchronous Partitioned Dual Decomposition (\asynch_name/) 
distributed algorithm}
\label{sec:AsynPBDD}

In this section we present an asynchronous partitioned distributed algorithm, 
and prove its convergence with high probability. This algorithm can be interpreted as 
an extension of the \synch_name/ distributed algorithm.

We consider an asynchronous protocol where each node has its own
concept of time defined by a local timer, which randomly and independently of
the other nodes triggers when to awake itself.
Each node is in an \emph{idle} mode, wherein it continuously receives messages
from neighboring nodes, until it is triggered either by the local timer or by a
message from neighboring nodes. When a trigger occurs, it switches into an
\emph{awake} mode in which it updates its local variables and possibly transmits
the updated information to its neighbors.
The timer is modeled by means of a local clock
$\tau_i\in\real_{\ge 0}$ and a randomly generated waiting time $T_i$. The timer
triggers the node when $\tau_i = T_i$, so that the node switches
to the awake mode and, after running the local computation, resets $\tau_i = 0$
and extracts a new realization of $T_i$.
We make the following assumption on the local waiting times $T_i$.
\begin{assumption}[Exponential i.i.d. local timers]
	The waiting times between consecutive triggering events are i.i.d. random variables 
	with same exponential distribution.~\oprocend  
\label{ass:timers}
\end{assumption}

Informally, the asynchronous distributed optimization algorithm is as follows.
When a node $i$ is in idle, it continuously receives messages from awake
neighbors. If the local timer $\tau_i$ triggers or new dual variables
$\lambda_i^{(j,i)}$, $\lambda_j^{(j,i)}$ are received, it wakes up.
When node $i$ wakes up, it updates and broadcasts its 
primal variable $y^{(i)} = \big( x_i^{(i)},\{x_j^{(i)}\}_{j\in\nbrs_i}\big)$,
computed through a local constrained minimization.
Moreover, if the transition was due to the local timer triggering, 
then it also updates and broadcasts its local dual variables 
$\lambda_i^{(i,j)}$ and $\lambda_j^{(i,j)}$, $j\in\nbrs_i$.
Since there is no global iteration counter, we highlight the difference 
between updated and not updated values during the ``awake'' phase, 
by means of a ``$^+$'' superscript symbol, 
e.g., we denote the updated primal variable as $x_i^{(i)+}$.

We want to stress some important aspects of the idle/awake cycle. First, these
two phases are regulated by local timers and local information exchange,
without the need of any central clock. Second, we assume that the computation 
in idle takes a negligible time compared to the one performed in the awake phase. 
Moreover, a constant, local step-size $\alpha_i$ is used in the ascent step,
which can  be initialized by means of local exchange of information between
neighboring nodes.
Finally, we point out that each agent uses the most updated values that are 
locally available to perform every computation.

The \asynch_name/ distributed algorithm is formally described in the following table.

\begin{algorithm}%
  \floatname{algorithm}{Distributed Algorithm}
  \begin{algorithmic}[0] 

    \Statex\textbf{Processor states}:
    $( x_i^{(i)}, \{ x_j^{(i)} \}_{j\in\nbrs_i}  )$ and
    $\{ \lambda_i^{(i,j)}, \lambda_j^{(i,j)}\}_{j\in \nbrs_i}$\smallskip

      \StatexIndent[0.5] Set $\tau_i = 0$ and get a realization $T_i$\smallskip
    
    \Statex \textbf{Evolution}:
    \Label \texttt{\textbf{\textit{IDLE:}}}

      \StatexIndent[0.4] \textsc{while:} $\tau_i < T_i$ \textsc{do}:\smallskip
      \StatexIndent[0.8] Receive $\lambda_i^{(j,i)}$, $\lambda_j^{(j,i)}$
      and/or $x_i^{(j)}$, $x_j^{(j)}$ from $j\in \nbrs_i$.
      
      \StatexIndent[0.8] \textsc{if}:\! dual variables are received go to \texttt{\textbf{\textit{AWAKE}}}.

     \StatexIndent[0.5] go to \texttt{\textbf{\textit{AWAKE}}}.

    \Label \texttt{\textbf{\textit{AWAKE:}}}
  
        \StatexIndent[0.5] Compute and broadcast
        \begin{small}
        \begin{align*}
	        & \big( x^{(i)+}_i,\{ x^{(i)+}_j\}_{j\in \nbrs_i}\big) = 
          \argmin_{ ( x_i,\{ x_j \}_{j\in\nbrs_i} ) \in X_i }
          \! f_{i}(x_i, \{ x_j \}_{j\in\nbrs_i} ) 
	          + x_i^\top \sum_{j \in \nbrs_i} 
	          \big(\lambda_i^{(i,j)}  - \lambda_i^{(j,i)} \big) 
	          +  \sum_{j \in \nbrs_i} x_j^\top \big( \lambda_j^{(i,j)}  - \lambda_j^{(j,i)} \big)
        \end{align*}
        \end{small}

      \StatexIndent[0.5] \textsc{if}:\! $\tau_i=T_i$ \textsc{then}:\! update and broadcast
      \begin{align}
      \label{eq:asynch_ascent}
      \begin{split}
        \lambda_i^{(i,j)+} & = \lambda_i^{(i,j)} + \alpha_i \big( x_i^{(i)+} - x_i^{(j)} \big) , \:\: \forall \, j \in \nbrs_{i},
        \\
        \lambda_j^{(i,j)+} & = \lambda_j^{(i,j)} + \alpha_i \big( x_j^{(i)+} - x_j^{(j)} \big) , \:\: \forall \, j \in \nbrs_{i},
      \end{split}
      \end{align}

      \StatexIndent[1] set $\tau_i=0$ and get a new realization $T_i$.\smallskip
      \StatexIndent[0.5] Go to \texttt{\textbf{\textit{IDLE}}}.

  \end{algorithmic}
  \caption{ \asynch_name/ } 
  \label{alg:asynch_pbdd}
\end{algorithm}

It is worth pointing out that being the algorithm asynchronous, for the 
analysis we need to carefully formalize the concept of algorithm iterations. 
We will use a nonnegative integer variable $t$ indexing a change in the whole 
state $\Lambda = [ \Lambda_1 \ldots \Lambda_n ]$ of the distributed algorithm. 
In particular, each triggering will induce an iteration of the distributed optimization 
algorithm and will be indexed with $t$.
We want to stress that this (integer) variable $t$ does not need to be known by the agents. 
That is, this timer is not a common clock and is only introduced for the sake of analysis.

\begin{theorem}
  Let Assumptions~\ref{ass:strong_convexity}, \ref{ass:convex_constraints} and~\ref{ass:slater}
  hold true.
  Let the timers $\tau_i$ satisfy Assumption~\ref{ass:timers} and step-sizes $\alpha_i$ 
  be constant and such that $0 < \alpha_i \le 1/L_i$, with 
  \begin{align}
    L_i = \sqrt{ 2 \textstyle\sum\limits_{j\in \nbrs_i}
      \Big(\frac{1}{\sigma_i} + \frac{1}{\sigma_j}\Big)^2}, \quad \forall\, i \in \until{n}.
  \label{eq:L_i_asynch_definition}
  \end{align} 
  Then, the random sequence $\{ \Lambda_1 (t), \ldots, \Lambda_n (t) \}$ generated by
  the \asynch_name/ (Distributed Algorithm~\ref{alg:asynch_pbdd}), %
  converges with high probability in objective value to the optimal cost
  $f^\star$ of problem \eqref{eq:partitioned_problem}, i.e.,
  for any $\varepsilon \in (0, q_{ 0 } )$, with 
  $q_{ 0 } :=  q(\Lambda(0))$,
  and 
  target confidence $0 < \rho < 1$, there exists
  $\bar t(\varepsilon,\rho)$ such that for all $t \ge \bar t(\varepsilon,\rho)$ it holds:
$\Pr \Big( \big| q(\Lambda(t)) - f^\star \big| \le \varepsilon  \Big) \ge 1 - \rho.    $
  \label{thm:asynch_convergence}
\end{theorem}

\begin{proof}
  Our proof strategy is based on showing that the iterations of the asynchronous
  distributed algorithm can be written as the iterations of an ad-hoc version of
  the coordinate method~\cite{richtarik2014iteration}, applied to 
  the dual problem~\eqref{eq:dual_problem}.

  Let the optimization variable $\Lambda$ be partitioned in $n$ blocks 
  $[\Lambda_1,\ldots,\Lambda_n]$ as in~\eqref{eq:lambda_i_definition},
  then a coordinate approach consists in an iterative scheme in which only a 
  block-per-iteration, say $\Lambda_{i_t}$ at time $t$, of the entire optimization 
  variable $\Lambda$ is updated at time $t$, while all the other components 
  $\Lambda_j$ with $j \in \until{n} \setminus \{i_t\}$ stay unchanged.
  Formally, a coordinate iteration can be summarized as
  \begin{align}
  \begin{split}
    \Lambda_{i_t}(t+1) & = \Lambda_{i_t} (t) + \nabla_{\Lambda_{i_t}} q( \Lambda(t) )
    \\
    \Lambda_{j}(t+1) & = \Lambda_{j} (t), \hspace{2.4cm} j \neq i_t.
  \end{split}
  \label{eq:coord_ascent}
  \end{align}

  In the following, we show that the \asynch_name/ distributed algorithm
  implements~\eqref{eq:coord_ascent} with a uniform random selection of 
  the blocks.
  Since the timers $\tau_i$ trigger independently according to the same 
  exponential distribution, then from an external, global perspective, the induced 
  awaking process of the nodes corresponds to the following: only one node per 
  iteration, say $i_t$, wakes up randomly, uniformly and independently from previous 
  iterations.
  Thus, each triggering, which induces an \emph{iteration} of the distributed
  optimization algorithm and is indexed with $t$, corresponds to the (uniform)
  selection of a node in $\until{n}$ that becomes awake.

  Next we show by induction that if each node $i$ has an updated 
  version of the neighboring variables before it gets awake, then the same holds
  after the update. 
  When node $i$ wakes up, it uses for its update its own primal variables
  $x_{i}^{(i)}$ and $x_{j}^{(i)}$, $j\in\nbrs_i$, which are clearly updated
  since $i$ is the one modifying them. Moreover, node $i$ uses also
  $x_{i}^{(j)}$ and $x_{j}^{(j)}$, $j\in\nbrs_i$, which are received by
  neighboring nodes $j\in\nbrs_i$. These variables are updated by $j$ if itself
  or one of its neighbors becomes awake. In both cases node $j$ sends the
  updated variable to its neighbors (which include node $i$).
  An analogous argument holds for the dual variables. 

  Thanks to the argument just shown and by noticing
  that $\lambda_{i_t}^{({i_t},j)}$ and $\lambda_{j}^{({i_t},j)}$, $j\in\nbrs_{i_t}$ 
  are the components of $\Lambda_{i_t}$, we have that step~\eqref{eq:asynch_ascent}
  corresponds to step~\eqref{eq:coord_ascent} with $i_t$ randomly uniformly 
  distributed over $\until{n}$.
  Finally, recalling that (i) the cost function $q(\Lambda)$ of 
  problem~\eqref{eq:dual_problem} has block-coordinate Lipschitz 
  continuous gradient with respect to the blocks $\Lambda_i$ 
  (see proof of Theorem~\ref{thm:synch_convergence}) and (ii) 
  the step-sizes $\alpha_i$ are constant and such that
  $0 < \alpha_i \le 1/L_i$ with $L_i$ in~\eqref{eq:L_i_asynch_definition},
  we can invoke~\cite[Theorem~5]{richtarik2014iteration}
  to conclude that the coordinate method~\eqref{eq:coord_ascent} 
  (and equivalently the \asynch_name/ distributed algorithm) converges with 
  high probability to the optimal cost $q^\star$ of problem~\eqref{eq:dual_problem}.
  Recalling that strong duality between problems~\eqref{eq:partitioned_problem} 
  and~\eqref{eq:dual_problem} holds (see proof of Theorem~\ref{thm:synch_convergence}),
  then $q^\star = f^\star$, and the proof follows.
\end{proof}

\begin{remark}
  As highlighted in Theorem~\ref{thm:asynch_convergence} in order to 
  set the local step-sizes $\alpha_i$, each node $i$ should know the convexity
  parameter $\sigma_j$ of its neighbors but, differently from the synchronous case 
  (cf. condition~\eqref{eq:L_i_synch_definition}),
  does not need to know the total number of agents $n$ in 
  the network.~\oprocend
  \label{rem:asynch_local_stepsizes}
\end{remark}

To conclude this section, we notice that the asynchronous model employed in our distributed 
algorithm can be generalized. In fact, in the considered model timers are drawn from
a common exponential distribution, while independent and completely uncoordinated rules might
be more desiderable. This generalization is currently under investigation.

\section{Numerical Simulations}
\label{sec:simulations}
In this section we provide a numerical example showing the effectiveness of the 
proposed techniques.
We test the proposed distributed algorithms on a quadratic program 
enjoying the partitioned structure described in the previous sections.
Specifically, we consider a network of $n = 100$ agents communicating 
according to an undirected connected Erd\H{o}s-R\'enyi random graph 
$\GG$ with parameter $p=0.2$. 
Thus, letting $\big( x_i, \big\{ x_j \big\}_{j \in \nbrs_i} \big)$ denote a 
column vector, we consider the following partitioned optimization problem
\begin{small}
\begin{align}
\begin{split}
  \min_{x} \: & \:  
    \sum_{i=1}^n \big( x_i, \big\{ x_j \big\}_{j \in \nbrs_i} \big)^\top Q_i 
    \big( x_i, \big\{ x_j \big\}_{j \in \nbrs_i} \big) 
    + r_i^\top \! \big( x_i, \big\{ x_j \big\}_{j \in \nbrs_i} \big)
  \\[1.2ex]
  \subj \: & \: A_i \big( x_i, \big\{ x_j \big\}_{j \in \nbrs_i} \big) \preceq b_i, \hspace{0.5cm} i\in\until{n},
\end{split}
\label{eq:partitioned_QP}
\end{align}
\end{small}
where each $x_i \in \real^{m_i}$ and $m_i$ is uniformly drawn from $\{1,2,3,4\}$.
This optimization problem has the same partitioned structure discussed in
Section~\ref{sec:dual_dec_setup}. In particular, we have quadratic cost
functions $f_i ( x_i,  \{ x_j \}_{j \in \nbrs_i} )$ and linear 
constraints $X_i = \setdef{ ( x_i, \big\{ x_j \big\}_{j \in \nbrs_i} ) }{A_i ( x_i, \big\{ x_j \big\}_{j \in \nbrs_i} ) \preceq b_i}$.
The matrices $Q_i$ are positive definite with eigenvalues uniformly generated in
$[1,20]$, while the vectors $r_i$ have entries randomly generated in
$[0,100]$. Moreover, each pair $A_i$, $b_i$ describes a linear constraint
  having a number of rows uniformly drawn from $\{1,2\}$.  Each $A_i$ has
entries normally distributed with zero mean and unitary variance, while $b_i$
are suitably generated to always obtain feasible linear constraints.
For all $i\in\until{n}$, we use constant step-sizes $\alpha_i = L_i$  with $L_i$ computed 
as in~\eqref{eq:L_i_synch_definition} for the synchronous algorithm and as
in~\eqref{eq:L_i_asynch_definition} for the asynchronous case.
All the dual variables are initialized to zero.

In Figure~\ref{fig:synch_cost} we show the convergence rate of the synchronous
distributed algorithm by plotting the difference between the dual cost
$q(\Lambda(t))$ at each iteration $t$ and the optimal value $q^\star = f^\star$
of problem~\eqref{eq:partitioned_QP}.
\begin{figure}[!htpb]
  \centering
  \includegraphics[scale=0.87]{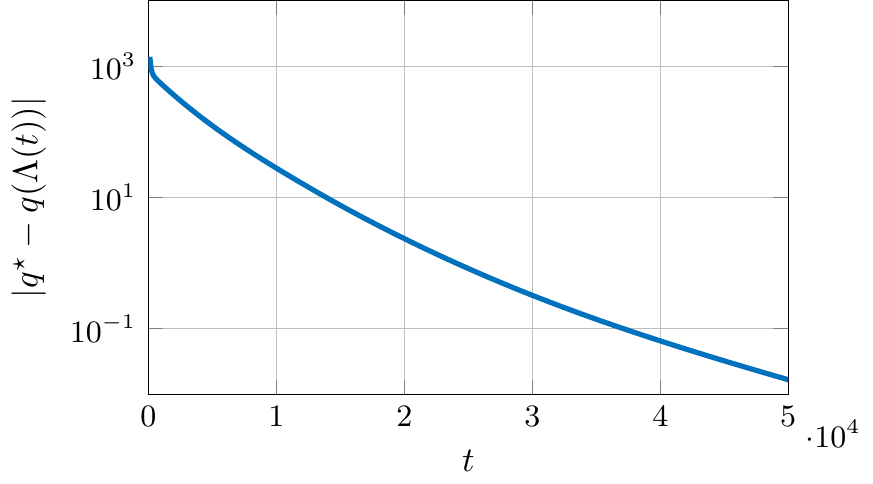}
  \caption{
    Evolution of the cost error for the synchronous distributed algorithm.
    }
  \label{fig:synch_cost} 
\end{figure}

In Figure~\ref{fig:synch_xx} we show the evolution of the difference 
between the generated primal sequence $\{ x_1^{(1)}(t),\ldots, x_n^{(n)}(t) \}$ 
and the (unique) optimal primal solution $x^\star$.
\begin{figure}[!htpb]
  \centering
  \includegraphics[scale=0.87]{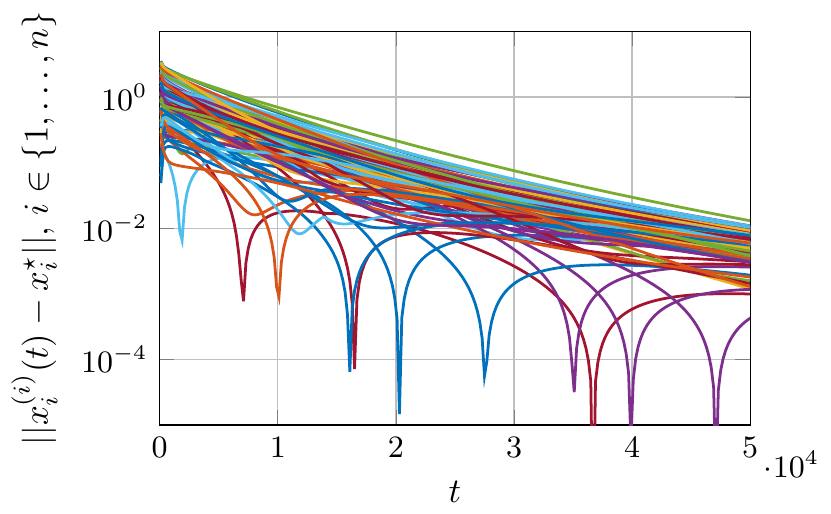}
  \caption{
    Evolution of the error on primal variables $x_i^{(i)}$, $i\in\until{n}$, for 
    the synchronous distributed algorithm.
    }
  \label{fig:synch_xx} 
\end{figure}

In Figure~\ref{fig:synch_agreement} we show the disagreement 
on the primal variable $x_2$ between neighboring nodes $\nbrs_2\cup  \{2\} $.
In particular, we plot the norm of $x_2^{(2)}(t)-x_2^{(j)}(t)$, for all $j\in\nbrs_2\cup  \{2\} $.
\begin{figure}[!htpb]
  \centering
  \includegraphics[scale=0.87]{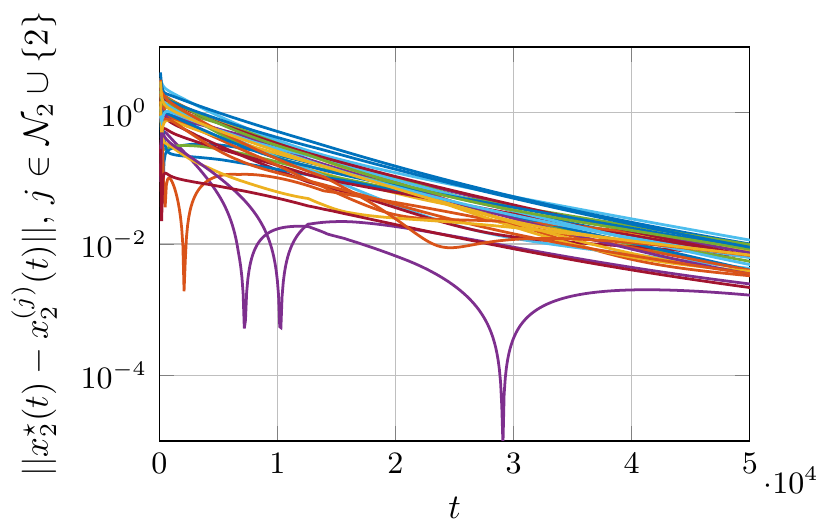}
  \caption{
    Evolution of the disagreement on $x_2$ between agents $2$ and its neighbors $j\in\nbrs_2$ %
    for the synchronous distributed algorithm.
    }
  \label{fig:synch_agreement} 
\end{figure}

Finally, in Figure~\ref{fig:asynch_cost} we show the convergence rate for the 
\asynch_name/ distributed algorithm.
Since we are dealing with an asynchronous algorithm, we normalize the iteration 
counter $t$ with respect to the total number of agents $n$.
It is worth noting the cost evolution is not monotone as expected 
for the class of randomized algorithms.
\begin{figure}
  \centering
  \includegraphics[scale=0.87]{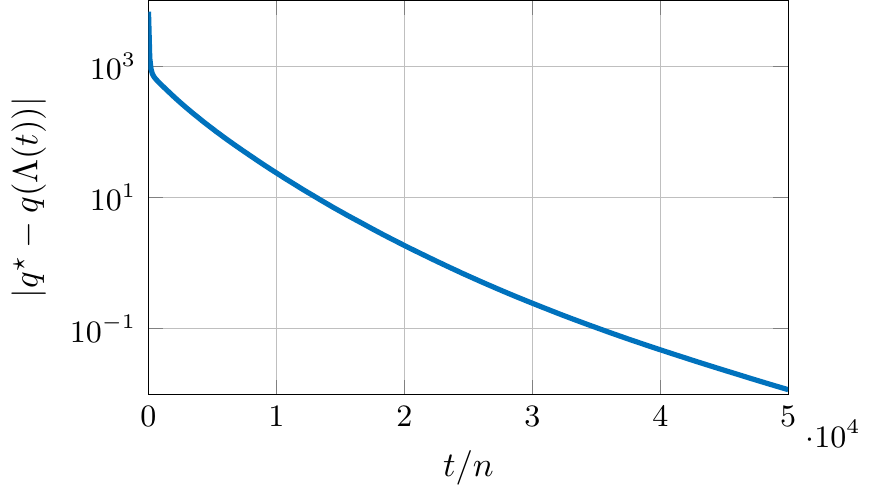}
  \caption{
    Evolution of the cost error for the asynchronous distributed algorithm.
    }
  \label{fig:asynch_cost} 
\end{figure}

\section{Conclusions}
\label{sec:conclusions}
In this paper we have proposed a synchronous and an asynchronous distributed
optimization algorithms, based on dual decomposition, for a novel partitioned
distributed optimization framework. In this framework each node in the network
is assigned a local state, objective function and constraint. The objective
function and the constraints only depend on the node state and on its neighbors'
states.
This scenario includes several interesting problems as network utility
maximization and resource allocation, static state estimation in power networks,
localization in wireless networks, and map building in robotic networks. 
The proposed algorithms are distributed and scalable and are shown to be
convergent under standard assumptions on the cost functions and on the
constraints sets.

\begin{small}
\bibliographystyle{IEEEtran} 
\bibliography{alias,partition_based} 
\vspace*{-1cm}
\end{small}

\end{document}